\documentclass[11pt,a4paper]{article}

\usepackage[utf8]{inputenc}
\usepackage[T1]{fontenc}
\usepackage[english]{babel}

\usepackage{amsmath,amssymb,amsthm,mathtools}
\usepackage{bm}

\usepackage{geometry}
\usepackage{graphicx}
\usepackage{subcaption}

\usepackage{xcolor}
\usepackage{float}
\usepackage[section]{placeins}

\usepackage{csquotes}
\usepackage{hyperref}
\usepackage[nameinlink,capitalize]{cleveref}

\usepackage{siunitx}
\usepackage{icomma}

\title{\textbf{Harmonic rigidity at fixed spectral gap in one dimension}}

\author{Arseny Pantsialei\thanks{e-mail: wselend@gmail.com} \\
\emph{Institute of Physics, Maria Curie-Sk\l{}odowska University, 20-031 Lublin, Poland}}

\newtheorem{theorem}{Theorem}
\newtheorem{lemma}{Lemma}
\newtheorem{remark}{Remark}

\DeclareMathOperator{\Var}{Var}

\begin{document}
\date{}
\maketitle

\begin{abstract}
We resolve the static isoperimetric problem underlying the Mandelstam-Tamm limit: among 1D confining potentials at fixed gap $\Delta$, the 
harmonic trap uniquely maximizes $\mathrm{Var}_0(x)$, yielding the exact geometric QSL $g_{xx} \le (2m\Delta)^{-1}$ with an iff criterion. 
Beyond the extremum we prove quantitative rigidity (TRK-tail and structural $L^2$ control), extend to magnetic settings (longitudinal iff; 
transverse guiding-center bounds), and note applications to static polarizability, quantum-metric limits, and trap benchmarking.
\end{abstract}

\section{Introduction}\label{sec:intro}

Quantum speed limits (QSLs) of the Mandelstam-Tamm (MT) type connect the geometry of 
state space to energy scales \cite{MandelstamTamm1945,AnandanAharonov1990,DeffnerCampbell2017}. 
For one-parameter shifts $x \mapsto x+\lambda$, the $xx$-component of the Fubini-Study metric is 
$g_{xx}=\Var_{0}(x)/\hbar^{2}$ \cite{ProvostVallee1980,BraunsteinCaves1994}. This leads to a natural 
{static isoperimetric} question: with a fixed spectral gap $\Delta=E_{1}-E_{0}$, how large can the
 geometric sensitivity $\Var_{0}(x)$ be across confining potentials $V$? The answer provides a firm benchmark
  for trap anharmonicity and a static analogue of a QSL for the quantum metric.

Historically, the tools that are enough for a {non-sharp}
 upper bound are well known: spectral decomposition, double
  commutators, and the Thomas-Reiche-Kuhn (TRK) $f$-sum rule 
  for the coordinate \cite{SakuraiQM,ReedSimonII}. In this sense, the ''skeleton'' of the inequality is folklore.
   However, two hard points remain that the standard approach \emph{does not} cover. These are (i) \emph{sharpness} with a 
   \emph{full if-and-only-if} equality criterion, and (ii) \emph{quantitative rigidity} (stability): how exactly the deficit 
   to the bound controls the spectral ''tail'' and the structural anharmonicity of the potential. This work closes both of these
    gaps.

Our main 1D result is simple and sharp: for all confining $V\in C^{2}(\mathbb{R})$ at fixed $\Delta$,
\[
\Var_{0}(x) \le \frac{\hbar^{2}}{2m\Delta},
\]
and {equality holds if and only if} $V(x)=\tfrac12 m\omega^{2}x^{2}+C$ with $\hbar\omega=\Delta$.
 In metric form this gives an {exact static analogue} of the MT limit: \(g_{xx}\le (2m\Delta)^{-1}\) with the same iff criterion. 
 The proof uses the ''active gap'' and an exactly sharp {master} bound for any self-adjoint observable $A$:
\[
\Var_{0}(A) \le \frac{1}{2\Delta_A}\langle [A,[H,A]]\rangle_0,\qquad
\Delta_A:=\min\{E_n-E_0:\ A_{n0}\neq0\},
\]
with a {clear equality condition}: $A\psi_0$ lies entirely in the active subspace (see §\ref{sec:setup}).
 For $A=x$, this together with TRK gives the sharp bound and the iff statement.

Beyond the extremal value, we establish quantitative rigidity (see §\ref{sec:rigidity}): the deficit
\[
\varepsilon:=\frac{\hbar^{2}}{2m\Delta}-\Var_{0}(x) \ge 0
\]
controls linearly (i) the TRK ''tail'' via the \emph{second gap} $\Gamma=E_{2}-E_{1}$ with 
sharp constants, and (ii) the structural anharmonicity via the $L^{2}(\rho_{0})$ 
norm of the deviation of the force $V'(x)$ from $m\omega^{2}x$. With an upper ''window'' on the energies in place, we obtain
 a direct lower bound.
Namely, $\varepsilon\gtrsim \|V'-m\omega^{2}x\|_{L^{2}(\rho_{0})}$ with an explicit constant
(formulas D1-D3).
 To the best of our knowledge, such an explicit stability picture for a fixed 
active gap in 1D has not been stated before.

The picture extends naturally. In many dimensions we get the same sharp bounds for the projections $x_u$ and an iff 
criterion of being ''quadratic along $u$'' (see §\ref{sec:dgt1-nofield}). For magnetic systems, we split the motion into
 parts along and across the field. Along the direction $\hat{\mathbf b}$ we have the exact identity
\[
[H,[H,x_{\parallel}]]=(\hbar^{2}/m)\partial_{\parallel}V,
\]
which gives the same sharp bound and a full iff statement even for inhomogeneous fields with a fixed direction 
(see \Cref{thm:iff-par,thm:iff-par-inhomB}). Across the field we use the guiding center $R_u$. Here we derive an exact 
TRK formula via the projected Hessian of $V$ and prove transverse analogues of D1-D2 (see §\ref{sec:magnetic-transverse}). 
All results are gathered and unified by the master inequality for a general observable $A$.

Within the MT/Fubini-Study viewpoint, we will read ''speed limits'' as metric constraints that become operational once 
the energy variance (or the FS metric) is known 
\cite{MandelstamTamm1945,AnandanAharonov1990,DeffnerCampbell2017,ProvostVallee1980,BraunsteinCaves1994}. In one-dimensional 
confining traps, a natural testbed is provided by the
center of mass ''rigid'' mode.
Kohn's theorem and its harmonic variant show that the collective frequency is
protected and that oscillator-like transported states persist even with
interactions \cite{Kohn1961,Dobson1994HPT,LaiPan2016SciRep}.
 This aligns with the geometric picture, since isolating the center 
of mass simplifies geodesic estimates and ties metric bounds directly to the external potential.

On the response side, static polarizabilities satisfy strict sum rules and clean scaling laws that constrain spectra and dipole matrix elements. Beyond the textbook 
TRK rule \cite{SakuraiQM}, recent density-based formulations clarify length scaling across broad classes of systems and its transfer from single-particle models to 
many-electron atoms and molecules \cite{Szabo2022PRL,Goger2024JCTC,Cheng2024PRA,Summa2023JCTC}. These inputs can be folded into metric bounds to yield operational 
constraints for concrete traps.

Experimentally, such speed/geometry bounds matter in ion traps, optomechanics, and atomic metrology, where preparation times and evolution rates set sensitivities 
and clock performance \cite{Leibfried2003RMP,Aspelmeyer2014RMP,PezzeSmerzi2018RMP}. This motivates \emph{potential-sensitive} QSL formulations tailored to common 
one-axis traps and effective 1D models, in the spirit of \cite{DeffnerCampbell2017}.

We work throughout within a standard functional-analytic framework (self-adjoint $H$, discrete spectrum, positivity of $\psi_{0}$) \cite{ReedSimonII}. 
For background on dynamical QSLs and geometric formulations see
\cite{MandelstamTamm1945,AnandanAharonov1990,DeffnerCampbell2017,ProvostVallee1980,BraunsteinCaves1994}.
For TRK and classic sum rules see \cite{SakuraiQM}.
For separability of the harmonic mode and Kohn's theorem in magnetic fields see
\cite{Kohn1961,Dobson1994HPT}.
 In summary 
(as detailed above), our contributions are: (i) a \emph{sharp} static QSL with a \emph{full iff} equality test; (ii) quantitative rigidity with explicit linear 
constants via the second gap $\Gamma$ and the TRK tail; (iii) \emph{structural} rigidity through $V'-m\omega^{2}x$; and (iv) longitudinal magnetic iff results and 
transverse TRK identities—unified by a single master inequality.

\section{Setup and basic facts}\label{sec:setup}

Consider the one-dimensional Schrödinger Hamiltonian
\begin{equation}\label{eq:H}
H=\frac{p^2}{2m}+V(x),
\end{equation}
where $m>0$ and $V\in C^{2}(\mathbb{R})$ is a confining potential (e.g., $\lim_{|x|\to\infty}V(x)=+\infty$). This ensures a 
discrete spectrum
\(
E_0<E_1\le E_2\le\ldots
\)
and a non-degenerate ground state $\psi_0>0$. For any operator $A$ we write
\[
\langle A\rangle_0:=\langle \psi_0, A \psi_0\rangle,\qquad
\Var_0(A):=\langle A^2\rangle_0-\langle A\rangle_0^2 .
\]

Let $\{\psi_n\}_{n\ge0}$ be the eigenfunctions of $H$: $H\psi_n=E_n\psi_n$, $\langle\psi_n,\psi_m\rangle=\delta_{nm}$. Write
\[
\Delta:=E_1-E_0>0 \quad\text{(first gap)},\qquad
\Gamma:=E_2-E_1\ge 0 \quad\text{(second gap)}.
\]
For convenience set $A_{n0}:=\langle\psi_n,A \psi_0\rangle$. In connection with the Thomas-Reiche-Kuhn sum rule,
 it is useful to define the ''TRK tail fraction''
\begin{equation}\label{eq:etaTRK}
\eta_{\mathrm{TRK}}
:=1-\frac{\Delta |x_{10}|^2}{S}\in[0,1],\qquad
S:=\sum_{n>0}(E_n-E_0) |x_{n0}|^2 ,
\end{equation}
where $x_{n0}:=\langle\psi_n,x \psi_0\rangle$.

Directly computing the commutators for \eqref{eq:H} gives
\begin{equation}\label{eq:double-comm}
[H,x]=-\frac{i\hbar}{m} p,\qquad
[x,[H,x]]=\frac{\hbar^2}{m},\qquad
[H,[H,x]]=\frac{\hbar^2}{m} V'(x).
\end{equation}
From this we immediately get the standard $f$-sum rule (TRK):
\begin{equation}\label{eq:TRK}
\sum_{n>0}(E_n-E_0) |x_{n0}|^2
=\frac12 \langle [x,[H,x]]\rangle_0
=\frac{\hbar^2}{2m}.
\end{equation}

Finally, we record a general ''sharp'' bound for any self-adjoint observable $A$. Define the active spectral gap

\begin{equation}\label{eq:DeltaA}
\Delta_A := \min\{E_n-E_0: A_{n0}\neq 0\},
\end{equation}
and the quantity
\begin{equation}\label{eq:SA}
S_A := \tfrac12\langle [A,[H,A]]\rangle_0 = \sum_{n>0}(E_n-E_0)|A_{n0}|^2 .
\end{equation}
Then the ground-state variance satisfies
\begin{equation}\label{eq:masterA}
\Var_0(A)=\sum_{n>0}|A_{n0}|^2
\le \frac{S_A}{\Delta_A}
= \frac{1}{2\Delta_A}\big\langle [A,[H,A]]\big\rangle_0 .
\end{equation}

From \eqref{eq:SA} we have
\[
2S_A=\sum_{n>0}2(E_n-E_0)|A_{n0}|^2.
\]
Since for all $n$ with $A_{n0}\neq 0$ we have $(E_n-E_0)\ge \Delta_A$, it follows that

\[
\sum_{n>0}|A_{n0}|^2
 \le  \sum_{n>0}\frac{E_n-E_0}{\Delta_A} |A_{n0}|^2
 =  \frac{S_A}{\Delta_A},
\]
which gives \eqref{eq:masterA}.

{{The inequality \eqref{eq:masterA} is completely general, but a universal sharp statement at fixed spectral gap
requires that the ''TRK constant'' $S_A$ be controlled independently of the potential.
In the 1D Schr\"odinger case $H=p^2/(2m)+V(x)$ this is the reason why the coordinate $A=x$ is special:
\(
[x,[H,x]]=\hbar^2/m
\)
is a constant, hence $S_x=\hbar^2/(2m)$ is potential-independent.
By contrast, for a multiplicative observable $A=f(x)$ one has
\(
[f(x),[H,f(x)]]=(\hbar^2/m)\,(f'(x))^2
\)
and therefore
\(
S_{f}=\frac{\hbar^2}{2m}\langle (f'(x))^2\rangle_0,
\)
which depends on the ground-state density (hence on $V$). Consequently, fixing a gap alone does not determine $S_f$,
and no universal sharp fixed-gap extremizer should be expected for general $f(x)$ without additional input controlling $S_A$.}}

If $V$ is even, then $\psi_0$ is even, $\psi_1$ is odd, and the matrix element $x_{10}\neq 0$. Hence for $A=x$ the active 
gap $\Delta_A$ equals the first gap $\Delta$. In many dimensions the same holds for $x_u$ when $V$ is symmetric under the 
reflection $x_u\mapsto -x_u$.

The equality criterion is as follows. In \eqref{eq:masterA} equality holds if and only if the vector $A\psi_0$ lies entirely
 in the eigenspace at energy $E_0+\Delta_A$, i.e., $A\psi_0\in\mathrm{Ran}P_{E_0+\Delta_A}$. In particular, if the active 
 level is non-degenerate, then $A\psi_0$ is proportional to the first ''active'' eigenfunction 
 $\psi_{1_A}$ with $E_{1_A}-E_0=\Delta_A$.

A full justification of \eqref{eq:double-comm}-\eqref{eq:TRK} (operator domains, closure of commutators, and spectral expansion)
 is given in Appendix~A. See also \cite{SakuraiQM}.

Finally, the quantum metric component for the ground state under the shift $x\mapsto x+\lambda$ is
\begin{equation}\label{eq:qmetric}
g_{xx}=\frac{\Var_0(x)}{\hbar^2},
\end{equation}
so any bounds on $\Var_0(x)$ are equivalent to sharp bounds on $g_{xx}$ at a fixed spectral gap.

\section{Main result: extremizing the variance at fixed gap}
\label{sec:theoremA}

Let $H=\frac{p^2}{2m}+V(x)$ with $V\in C^2(\mathbb R)$ confining. Then the spectrum is discrete
$E_0<E_1\le E_2\le\ldots$, and the ground state is non-degenerate with $\psi_0>0$ \cite{ReedSimonII}.
Set $\Delta:=E_1-E_0>0$ and $\rho_0:=|\psi_0|^2$. All domain issues
(correctness of commutators and expansions) are collected in App.~A.

Among all confining potentials
with a fixed first gap $\Delta=E_1-E_0$ one has
\begin{equation}\label{eq:mainbound}
\Var_0(x)\ \le\ \frac{\hbar^2}{2m\Delta},
\end{equation}
and equality in \eqref{eq:mainbound} holds if and only if

\begin{equation}
V(x)=\tfrac12m\omega^2 x^2+C,\qquad \hbar\omega=\Delta,
\end{equation}

with an arbitrary constant $C$. In particular,
\[
\sup_{V:E_1-E_0=\Delta}\Var_0(x)=\frac{\hbar^2}{2m\Delta},
\]
and, up to adding $C$, the unique extremizer is the harmonic oscillator.

\emph{Proof.}
The spectral expansion gives

\begin{equation}\label{eq:var-spectral}
\Var_0(x)=\sum_{n>0}|x_{n0}|^2,\qquad x_{n0}:=\langle\psi_n,x\psi_0\rangle .
\end{equation}

For the Hamiltonian \eqref{eq:H} the double commutator satisfies
$[x,[H,x]]=\hbar^2/m$ (see \eqref{eq:double-comm}), hence the standard TRK formula

\begin{equation}\label{eq:trk_local}
\sum_{n>0}(E_n-E_0)|x_{n0}|^2
=\tfrac12\langle [x,[H,x]]\rangle_0
=\frac{\hbar^2}{2m}.
\end{equation}

Since $E_n-E_0\ge\Delta$ for all $n>0$, from \eqref{eq:var-spectral}-\eqref{eq:trk_local} we obtain
\[
\Var_0(x)
=\sum_{n>0}\frac{E_n-E_0}{E_n-E_0}|x_{n0}|^2
\le \frac{1}{\Delta}\sum_{n>0}(E_n-E_0)|x_{n0}|^2
=\frac{\hbar^2}{2m\Delta},
\]

which is \eqref{eq:mainbound}. Equality can occur only if all terms with $E_n-E_0>\Delta$ vanish, i.e.
\[
x\psi_0\in\mathrm{span}\{\psi_1\}.
\]
Writing $x\psi_0=a\psi_1$ ($a\neq0$), we get $(H-E_0)x\psi_0=\Delta x\psi_0$ and therefore
\[
[H,[H,x]]\psi_0=(H-E_0)^2x\psi_0=\Delta^2 x\psi_0.
\]

On the other hand, \eqref{eq:double-comm} gives the operator identity $[H,[H,x]]=(\hbar^2/m) V'(x)$. Since $\psi_0>0$ (a.e.), 
it follows that the force is linear:
\[
V'(x)=\frac{m}{\hbar^2}\Delta^2 x=:m\omega^2 x,\qquad \omega=\Delta/\hbar.
\]
After integration $V(x)=\tfrac12 m\omega^2 x^2+C$. For such $V$ one has $\Var_0(x)=\hbar/(2m\omega)=\hbar^2/(2m\Delta)$, 
so the upper bound is attained, and no other potential (up to adding a constant $C$) can give equality. \hfill$\square$

We introduce equivalent equality criteria:
\begin{enumerate}\itemsep0.25em
\item $\Var_0(x)=\hbar^2/(2m\Delta)$;
\item $x \psi_0\in\mathrm{span}\{\psi_1\}$;
\item the TRK sum is saturated by the single transition $0\to1$, i.e.,
\[
\sum_{n>0}(E_n-E_0)|x_{n0}|^2
=\Delta |x_{10}|^2;
\]
\item $[H,[H,x]] \psi_0=\Delta^2 x\psi_0$ (and hence $V'(x)$ is linear).
\end{enumerate}

Indeed, (1)$\Rightarrow$(2): equality in the step replacing $(E_n-E_0)$ by $\Delta$ is only possible when $x_{n0}=0$ for 
all $n\ge2$. (2)$\Rightarrow$(3): then $(H-E_0)x\psi_0=\Delta x\psi_0$, and the TRK sum has the single term $n=1$. 
(3)$\Rightarrow$(4): as above, $[H,[H,x]] \psi_0=\Delta^2 x\psi_0$. (4)$\Rightarrow$(1): taking the inner product with 
$x\psi_0$ and using $[x,[H,x]]=\hbar^2/m$ gives $2\Delta \Var_0(x)=\hbar^2/m$.

All statements are invariant under adding a constant $C$ to $V$. In many dimensions, for the observable $A=u \cdot x$ 
one gets $\Var_0(u \cdot x)\le \hbar^2/(2m \Delta_u)$, and equality holds only when $V$ is quadratic along the direction 
$u$ (see §§\ref{sec:magnetic-parallel}, \ref{sec:magnetic-transverse}).

{{The bound \(\Var_0(x)\le \hbar^2/(2m\Delta)\) applies to any globally confining \(V\) with discrete spectrum.
However, in a deep symmetric double-well the first gap \(\Delta=E_1-E_0\) is an exponentially small tunnelling
splitting. Thus the right-hand side can become very large while \(\Var_0(x)\) stays of order of the well separation.
This is consistent with our rigidity picture. Near-saturation forces \(V\) to be close to the harmonic extremizer,
whereas a deep double-well is maximally far from harmonic. Thus in the tunnelling regime the inequality remains true
but may lose quantitative usefulness unless one fixes additional spectral/geometric data beyond \(\Delta\), or considers
an observable whose active gap is not set by the tunnelling splitting.}}

\section{Quantitative rigidity at fixed gap}
\label{sec:rigidity}

We study the deficit from the limiting value
\[
\varepsilon := \frac{\hbar^2}{2m \Delta}-\Var_0(x)\ \ge 0,
\]
and how the ''tail'' from higher states contributes to the position variance. Let $\{\psi_n\}_{n\ge0}$ be the eigenfunctions, 
$H\psi_n=E_n\psi_n$, and set $x_{n0}:=\langle\psi_n,x \psi_0\rangle$. Then
\[
S:=\sum_{n>0}(E_n-E_0) |x_{n0}|^2=\frac{\hbar^{2}}{2m}\quad\text{(TRK)},\qquad
\Delta:=E_1-E_0>0,\quad
\Gamma:=E_2-E_1\ge 0,
\]
and we define the ''tail'' and the TRK tail fraction by
\[
T:=\sum_{n\ge2}|x_{n0}|^2,\qquad
\eta_{\mathrm{TRK}}
:=1-\frac{\Delta |x_{10}|^2}{S}\in[0,1].
\]
All domain issues and the derivation of TRK are collected in App.~A.

It is convenient to fix a universal rigidity form that transfers unchanged to any self-adjoint observable $A$. 
This is done by the replacements
\(x_{n0} \to A_{n0}:=\langle\psi_n,A \psi_0\rangle\),
\(S \to S_A\) from \eqref{eq:SA},
\(\Delta \to \Delta_A\),
\(\Gamma \to \Gamma_A\). 
Here
\begin{equation}\label{eq:GammaA}
\Gamma_A := \min\{ (E_n-E_0)-\Delta_A: A_{n0}\neq0, E_n-E_0>\Delta_A \}.
\end{equation}

Let $w_{\mathrm{tail}}:=\sum_{E_n-E_0\ge \Delta_A+\Gamma_A}|A_{n0}|^2$ be the mass of $A\psi_0$ outside the active subspace. 
Then from \eqref{eq:SA}-\eqref{eq:masterA} it follows that
\begin{equation}\label{eq:stabilityA}
\frac{1}{2\Delta_A} \big\langle [A,[H,A]]\big\rangle_0 - \Var_0(A)
 \ge  \frac{\Gamma_A}{\Delta_A} w_{\mathrm{tail}}.
\end{equation}
In other words, if $\Var_0(A)$ is close to the right-hand side of \eqref{eq:masterA}, then the vector $A\psi_0$ is almost 
entirely concentrated in the active subspace, and moreover
\[
w_{\mathrm{tail}} \le \frac{\Delta_A}{\Gamma_A} \left(\frac{1}{2\Delta_A} \big\langle [A,[H,A]]\big\rangle_0-\Var_0(A)\right).
\]

\subsection{Statements and explicit constants}

We start with a lower bound on the deficit via the TRK tail. For any confining $V\in C^2$ with fixed $\Delta$ and $\Gamma$ 
we have
\begin{equation}\label{eq:D1}
\boxed{\quad
\varepsilon \ge \frac{S \Gamma}{\Delta(\Delta+\Gamma)} \eta_{\mathrm{TRK}}
 = \frac{\hbar^2}{2m} \frac{\Gamma}{\Delta(\Delta+\Gamma)} \eta_{\mathrm{TRK}} .
\quad}
\tag{D1}
\end{equation}

A small deficit suppresses both the total variance ''tail'' and its TRK share. In particular,
\begin{equation}\label{eq:D2-tail}
\boxed{\qquad
T=\sum_{n\ge2}|x_{n0}|^2 \le \frac{\Delta}{\Gamma} \varepsilon,
\qquad
\eta_{\mathrm{TRK}} \le \frac{\Delta(\Delta+\Gamma)}{S \Gamma} \varepsilon
 = \frac{2m \Delta(\Delta+\Gamma)}{\Gamma \hbar^2} \varepsilon .
\qquad}
\tag{D2}
\end{equation}
Hence $\varepsilon\to 0$ implies $T\to 0$ and $\eta_{\mathrm{TRK}}\to 0$ with explicit rates.

For structural control, it is convenient to measure the deviation of $V'$ from the ''harmonic'' force with frequency 
$\omega:=\Delta/\hbar$ in the $\rho_0=|\psi_0|^2$ metric:
\[
g(x):=V'(x)-m\omega^2 x,\qquad  g _{L^2(\rho_0)}^2=\int_{\mathbb R} |g(x)|^2 \rho_0(x) dx.
\]
Then
\begin{equation}\label{eq:D3}
\boxed{\quad
\sum_{n\ge2}|x_{n0}|^2 \le
\frac{\hbar^4}{m^2} \frac{ g _{L^2(\rho_0)}^2}{\bigl[\Gamma(2\Delta+\Gamma)\bigr]^2} .
\quad}
\tag{D3}
\end{equation}
In other words, if $V'(x)$ is close to $m\omega^2 x$ in $L^2(\rho_0)$, then the tail $T$ is small with an explicit constant.

Let us note that \eqref{eq:D3} controls the tail $T$ via $g$, but it does not give a universal lower bound of the form 
$\varepsilon\ge c  g ^2$ without extra UV control of the spectrum. Indeed, for fixed $g$, the contribution of levels with large 
$E_n$ can be arranged so that $(E_n-E_0)-\Delta$ stays arbitrarily small in the sum.

If, in addition, for all $n\ge2$ with $x_{n0}\neq0$ there is an upper ''window''
\(
E_n-E_0\le \Delta+\Lambda
\)
(i.e., active transitions lie below $\Delta+\Lambda$), then we have the stronger bound
\begin{equation}\label{eq:UV}
\boxed{\quad
\varepsilon \ge \frac{\hbar^4}{4m^2}\ \frac{ g _{L^2(\rho_0)}^2}{\Delta^3 (\Delta+\Lambda)} .
\quad}
\end{equation}
For fixed $\Delta,\Lambda$, small $ g $ then forces small $\varepsilon$ with an explicit constant.

Finally, introduce normalized weights
\[
f_{n0}:=\frac{2m}{\hbar^2} (E_n-E_0) |x_{n0}|^2,\qquad
\eta_{\mathrm{TRK}}=\sum_{n\ge2} f_{n0},
\]
and the quantity
\[
\tilde\eta:=1-\Delta\sum_{n>0}\frac{f_{n0}}{E_n-E_0}.
\]
Then there is a ''corridor'' between two extreme tail distributions:
\begin{equation}\label{eq:corridor-main}
\boxed{\quad \frac{\Gamma}{\Delta+\Gamma} \eta_{\rm TRK} \le \tilde\eta \le \eta_{\rm TRK} .\quad}
\end{equation}
The left edge is attained when the entire tail sits at level $n=2$. The right edge when the tail is ''infinitely far'' 
in energy.

\subsection{Proof of D1}\label{thm:D1}

From TRK we get an exact form of the deficit:
\begin{align}
\varepsilon
&=\frac{1}{\Delta}\sum_{n>0}(E_n-E_0) |x_{n0}|^2-\sum_{n>0}|x_{n0}|^2 \notag\\
&=\frac{1}{\Delta}\sum_{n>0}\bigl[(E_n-E_0)-\Delta\bigr]|x_{n0}|^2
=\frac{1}{\Delta}\sum_{n\ge2}\bigl[(E_n-E_0)-\Delta\bigr]|x_{n0}|^2.
\label{eq:eps-exact}
\end{align}
Since for $n\ge2$ we have $E_n-E_0\ge\Delta+\Gamma$, it follows that
\begin{equation}\label{eq:eps-ge}
\varepsilon \ge \frac{\Gamma}{\Delta}\sum_{n\ge2}|x_{n0}|^2
 = \frac{\Gamma}{\Delta}T.
\end{equation}
On the other hand, by TRK (see \eqref{eq:TRK})
\[
S=\Delta |x_{10}|^2+\sum_{n\ge2}(E_n-E_0)|x_{n0}|^2
 \ge \Delta |x_{10}|^2+(\Delta+\Gamma) T,
\]
whence
\begin{equation}\label{eq:T-upper}
T \le \frac{S-\Delta|x_{10}|^2}{\Delta+\Gamma}
 = \frac{S}{\Delta+\Gamma} \eta_{\mathrm{TRK}}.
\end{equation}
Combining the exact form \eqref{eq:eps-exact} as
\[
\varepsilon=\frac{1}{\Delta}\bigl(S-\Delta|x_{10}|^2-\Delta T\bigr)
=\frac{S}{\Delta} \eta_{\mathrm{TRK}}-T
\]
with \eqref{eq:T-upper}, we get
\[
\varepsilon \ge \frac{S}{\Delta} \eta_{\mathrm{TRK}}
-\frac{S}{\Delta+\Gamma} \eta_{\mathrm{TRK}}
=\frac{S \Gamma}{\Delta(\Delta+\Gamma)} \eta_{\mathrm{TRK}},
\]
which proves \eqref{eq:D1}.

\subsection{Proof of D2}\label{thm:D2}

The first part of \eqref{eq:D2-tail} follows directly from \eqref{eq:eps-ge}:
\(
T\le (\Delta/\Gamma) \varepsilon.
\)
For the second, use the exact identity
\(
\varepsilon=\frac{1}{\Delta}\bigl(S \eta_{\mathrm{TRK}}-\Delta T\bigr),
\)
so that
\[
\eta_{\mathrm{TRK}}
=\frac{\Delta}{S} (\varepsilon+T)
 \le \frac{\Delta}{S} \left(\varepsilon+\frac{\Delta}{\Gamma}\varepsilon\right)
=\frac{\Delta(\Delta+\Gamma)}{S \Gamma} \varepsilon,
\]
as required.

\subsection{Proof of D3}\label{thm:D3}
Consider the vector
\[
\Phi := \bigl((H-E_0)^2-\Delta^2\bigr) x \psi_0.
\]
Since $(H-E_0)^2\psi_0=0$, we have
\[
\Phi=\bigl[(H-E_0)^2,x\bigr]\psi_0-\Delta^2 x\psi_0
=\bigl([H,[H,x]]-\Delta^2 x\bigr)\psi_0.
\]
Using the identity $[H,[H,x]]=(\hbar^2/m) V'(x)$ and the definition
$g(x)=V'(x)-m\omega^2 x$ with $\omega=\Delta/\hbar$, we get
\begin{equation}\label{eq:Phi-right}
 \Phi^2
=\left( \frac{\hbar^2}{m} g(x) \psi_0 \right)^{2}
=\frac{\hbar^4}{m^2}\  g _{L^2(\rho_0)}^2.
\end{equation}
On the other hand, expansion in eigenstates yields
\begin{equation}\label{eq:Phi-left}
 \Phi ^2
=\sum_{n>0}\bigl((E_n-E_0)^2-\Delta^2\bigr)^{2} |x_{n0}|^2
 \ge \bigl[\Gamma(2\Delta+\Gamma)\bigr]^2 \sum_{n\ge2}|x_{n0}|^2,
\end{equation}
since for $n\ge2$ we have $E_n-E_0\ge\Delta+\Gamma$, and the function
$a\mapsto a^2-\Delta^2$ is increasing on $a\ge\Delta$.
Comparing \eqref{eq:Phi-right} and \eqref{eq:Phi-left} gives \eqref{eq:D3}. \hfill$\square$

\subsection{Proof of the upper-window strengthening}

Assume that for all $n\ge2$ with $x_{n0}\ne0$ we have $E_n-E_0\le\Delta+\Lambda$.
Then for such $n$,
\[
(E_n-E_0)-\Delta = \frac{(E_n-E_0)^2-\Delta^2}{(E_n-E_0)+\Delta}
 \ge \frac{(E_n-E_0)^2-\Delta^2}{2\Delta+\Lambda}.
\]
Plugging this into \eqref{eq:eps-exact} and using the Cauchy-Bunyakovsky inequality,
\begin{align*}
\varepsilon
&=\frac{1}{\Delta}\sum_{n\ge2}\bigl[(E_n-E_0)-\Delta\bigr]|x_{n0}|^2
\ \ge\ \frac{1}{\Delta(2\Delta+\Lambda)}\sum_{n\ge2}\bigl((E_n-E_0)^2-\Delta^2\bigr)|x_{n0}|^2\\
&\ge \frac{1}{\Delta(2\Delta+\Lambda)}\
\frac{\displaystyle\sum_{n\ge2}\bigl((E_n-E_0)^2-\Delta^2\bigr)^{2}|x_{n0}|^2}{ \displaystyle\max_{n\ge2}\bigl((E_n-E_0)^2-\Delta^2\bigr) }.
\end{align*}
Here $\max_{n\ge2}\bigl((E_n-E_0)^2-\Delta^2\bigr)\le(\Delta+\Lambda)^2-\Delta^2=\Lambda(2\Delta+\Lambda)$,
and the numerator equals $ \Phi ^2$ from \eqref{eq:Phi-right}-\eqref{eq:Phi-left}. Hence,
\[
\varepsilon \ge \frac{1}{\Delta(2\Delta+\Lambda)}
\frac{ \Phi ^2}{\Lambda(2\Delta+\Lambda)}
=\frac{1}{\Delta \Lambda (2\Delta+\Lambda)^2} \frac{\hbar^4}{m^2}  g _{L^2(\rho_0)}^2.
\]
Since $(2\Delta+\Lambda)^2\le 4\Delta(\Delta+\Lambda)$, we obtain the convenient (slightly weaker) explicit form \eqref{eq:UV}:
\[
\varepsilon \ge \frac{\hbar^4}{4m^2} \frac{ g _{L^2(\rho_0)}^2}{\Delta^{3}(\Delta+\Lambda)}.
\]
\hfill$\square$

\subsection{Remarks and corollaries}

The constants in \eqref{eq:D1}-\eqref{eq:D3} are sharp. For example, \eqref{eq:D1} becomes an equality if the whole tail 
sits on a single level with $E_2=E_1+\Gamma$, while levels with $n>2$ are ''sent to infinity'' (an idealized two-level tail).
 Thus the factor in front of $\eta_{\mathrm{TRK}}$ is optimal. The estimate \eqref{eq:D2-tail} is also optimal in order, 
 including in the limit $\Gamma\downarrow 0$.

The scaling matches all constants. Under the affine stretch $x\mapsto\lambda x$ and $V\mapsto V_\lambda(x):=V(\lambda x)$, 
the energy levels scale so that $\Delta\mapsto \lambda^2\Delta$ (with a corresponding rescale of $m$ or time units), 
and formulas \eqref{eq:D1}-\eqref{eq:D3} keep their form.

It is also useful to record a direct consequence for the static dipole polarizability. From the identity
\[
\alpha(0)=2\sum_{n>0}\frac{|x_{n0}|^2}{E_n-E_0}
\]
and from TRK we get
\[
\alpha(0) \le \frac{2}{\Delta}\sum_{n>0}|x_{n0}|^2
=\frac{2 \Var_0(x)}{\Delta}
 \le \frac{\hbar^2}{m \Delta^2},
\]
and equality is possible only for the harmonic potential. In terms of the quantum metric $g_{xx}=\Var_0(x)/\hbar^2$ this gives 
the sharp bound $g_{xx}\le (2m\Delta)^{-1}$. The same quantitative rigidity as in \eqref{eq:D1}-\eqref{eq:D3}.

In a uniform magnetic field all statements have verbatim analogues for the transverse coordinate / guiding center $R_u$ 
(see §\ref{sec:magnetic-transverse}). It suffices to make the replacements
\[
x \leadsto R_u,\qquad
S \leadsto S_{R_u}=\frac{\ell_B^4}{2}\Big\langle w^{ \top}(\nabla\nabla V)w\Big\rangle_0,\qquad
\Delta \leadsto \Delta_{R_u},\qquad
\Gamma \leadsto \Gamma_{R_u}.
\]
After these substitutions one obtains
\[
\varepsilon_u:=\frac{S_{R_u}}{\Delta_{R_u}}-\Var_0(R_u) \ge
\frac{S_{R_u} \Gamma_{R_u}}{\Delta_{R_u}(\Delta_{R_u}+\Gamma_{R_u})}\ \eta_{\mathrm{TRK}}^{(R_u)},
\qquad
\sum_{n\ge2}|\langle\psi_n,R_u \psi_0 \rangle|^2 \le \frac{\Delta_{R_u}}{\Gamma_{R_u}} \varepsilon_u,
\]
and the structural estimate \eqref{eq:D3} is rewritten with the replacement $g\mapsto w^{ \top}(\nabla\nabla V)w-\langle\cdot\rangle_0$.

Finally, the practical meaning of \eqref{eq:D1}-\eqref{eq:D3} is that they provide a gauge of anharmonicity at a fixed 
working gap $\Delta$. The deficit $\varepsilon$ quantitatively measures the trap's ''non-harmonicity'': it 
controls the spectral tail via $\eta_{\mathrm{TRK}}$ and the structural deviation of $V'$ via \eqref{eq:D3} and \eqref{eq:UV}. 
In particular, for a given $\Delta$ no anharmonic trap can exceed the limit $\Var_0(x)=\hbar^2/(2m\Delta)$, and small 
$\varepsilon$ certifies near-harmonic behavior: $T$ and $\eta_{\mathrm{TRK}}$ tend to zero as in \eqref{eq:D2-tail}, 
while the force deviation is bounded through \eqref{eq:D3} and improved by \eqref{eq:UV}.

\subsection{Multidimensional case $d>1$ without a magnetic field}

\label{sec:dgt1-nofield}

Let $H=\frac{p^2}{2m}+V(\mathbf x)$ on $\mathbb R^d$, with $V\in C^2$ confining; the ground state is non-degenerate, 
$\psi_0>0$, and the spectrum is discrete. For a unit vector $u\in\mathbb S^{d-1}$ set $x_u:=u \cdot \mathbf x$ and define
\[
x^{(u)}_{n0}:=\langle \psi_n, x_u \psi_0\rangle,\qquad
\Delta_u:=\min\{ E_n-E_0: x^{(u)}_{n0}\neq 0 \}  (\Delta_u\ge E_1-E_0).
\]
Then we have the sharp bound
\begin{equation}\label{eq:var-multi-nofield}
\Var_0(x_u) \le \frac{\hbar^2}{2m \Delta_u},
\end{equation}
and equality in \eqref{eq:var-multi-nofield} holds if and only if the potential is quadratic along $u$:
\begin{equation}\label{eq:V-quad-along-u}
V(\mathbf x)=\frac{m\omega^2}{2} x_u^2+W(x_u^\perp)+C,\qquad \hbar\omega=\Delta_u,
\end{equation}
where $x_u^\perp:=\mathbf x-(u \cdot \mathbf x) u$, $W$ is any confining function of the transverse variable only, 
and $C\in\mathbb R$.

The proof repeats the one-dimensional case. The TRK sum rule for $x_u$ gives
\[
\sum_{n>0}(E_n-E_0) |x^{(u)}_{n0}|^2=\frac{\hbar^2}{2m},
\]
and the standard ''Chebyshev'' step with the active gap $\Delta_u$ yields \eqref{eq:var-multi-nofield} at once. If 
equality holds, then as in $d=1$ the vector $x_u\psi_0$ lies in the active subspace: $x_u\psi_0\in\mathrm{span}\{\psi_{1,u}\}$, 
whence
\[
[H,[H,x_u]] \psi_0=(H-E_0)^2x_u\psi_0=\Delta_u^2 x_u\psi_0.
\]
Since $[H,[H,x_u]]=(\hbar^2/m) \partial_u V$, we get $\partial_u V(\mathbf x)=m\omega^2 x_u$ on the support of $\rho_0$, 
and integrating in $x_u$ at fixed $x_u^\perp$ gives \eqref{eq:V-quad-along-u}. The converse check is immediate by plugging 
into \eqref{eq:var-multi-nofield}.

\section{Uniform magnetic field: parallel direction}
\label{sec:magnetic-parallel}

Let $\mathbf B=B \hat{\mathbf b}$ be a uniform field, with $q\neq0$,
\[
H=\frac{\boldsymbol{\pi}^2}{2m}+V(\mathbf x),\qquad
\boldsymbol{\pi}:=\mathbf p-q \mathbf A(\mathbf x),\qquad
\nabla\times\mathbf A=\mathbf B,
\]
where $V\in C^2(\mathbb R^d)$ is confining. The ground state is non-degenerate, and $\psi_0>0$.

{{Throughout we work in a global Euclidean setting ($\mathbb{R}^d$) with a globally defined smooth vector potential $\mathbf A$,
so no gauge patching, bundle topology, or boundary/holonomy subtleties enter.
Extensions to manifolds or nontrivial $U(1)$ line bundles would require a reformulation in terms of globally defined
objects and a careful control of domains/boundary contributions. These are outside our present scope.}}

Define the parallel objects
\[
x_\parallel:=\hat{\mathbf b} \cdot \mathbf x,\qquad
\partial_\parallel:=\hat{\mathbf b} \cdot\nabla .
\]
(Domains and commutator validity are in App.~A.) We use the gauge-covariant commutators
\[
[x_i,\pi_j]=i\hbar \delta_{ij},\qquad
[\pi_i,\pi_j]=i\hbar q \varepsilon_{ijk}B_k,\qquad
[V(\mathbf x),\pi_i]=i\hbar \partial_i V(\mathbf x).
\]
Hence
\begin{equation}\label{eq:Hxpar}
[H,x_\parallel]
=\frac{1}{2m}[\pi^2,x_\parallel]
=- \frac{i\hbar}{m} \hat{\mathbf b} \cdot \boldsymbol{\pi},
\end{equation}
and therefore
\begin{align}
[H,[H,x_\parallel]]
&=- \frac{i\hbar}{m} [H, \hat{\mathbf b} \cdot \boldsymbol{\pi}]
=- \frac{i\hbar}{m}\Bigl(\frac{1}{2m}[\pi^2,\hat{\mathbf b} \cdot \boldsymbol{\pi}]
+[V, \hat{\mathbf b} \cdot \boldsymbol{\pi}]\Bigr).\label{eq:double-par-pre}
\end{align}
The first commutator in \eqref{eq:double-par-pre} vanishes, since
\[
[\pi_i, \hat{\mathbf b} \cdot \boldsymbol{\pi}]
=\hat b_j[\pi_i,\pi_j]
=i\hbar q \varepsilon_{ijk}\hat b_j B_k
=i\hbar q (\hat{\mathbf b}\times\mathbf B)_i=0
\]
(the vector $\hat{\mathbf b}$ is parallel to $\mathbf B$). For the second commutator we have
\(
[V, \hat{\mathbf b} \cdot \boldsymbol{\pi}]
=i\hbar \hat{\mathbf b} \cdot\nabla V
=i\hbar \partial_\parallel V.
\)

Thus we obtain the exact identity
\begin{equation}\label{eq:double-par}
\boxed{\quad
[H,[H,x_\parallel]]=\frac{\hbar^2}{m} \partial_\parallel V(\mathbf x).
\quad}
\end{equation}

From \eqref{eq:Hxpar} we also immediately get the TRK equality along the field:
\[
\sum_{n>0}(E_n-E_0) |\langle\psi_n, x_\parallel \psi_0\rangle|^2
=\tfrac12 \langle [x_\parallel,[H,x_\parallel]]\rangle_0
=\frac{\hbar^2}{2m},
\]
so the master bound \eqref{eq:masterA} for $A=x_\parallel$ takes the form
\begin{equation}\label{eq:var-par-bound}
\Var_0(x_\parallel)\ \le\ \frac{\hbar^2}{2m \Delta_\parallel},
\end{equation}
where 

\begin{equation}\label{eq:gap-par}
\Delta_\parallel:=\min\{E_n-E_0:\ \langle\psi_n, x_\parallel \psi_0\rangle\neq0\}.
\end{equation}

\begin{theorem}[Longitudinal \emph{iff} in a uniform field]\label{thm:iff-par}
The bound \eqref{eq:var-par-bound} holds, and equality is achieved \emph{if and only if}
\[
V(\mathbf x)=\tfrac12 m\omega^2 x_\parallel^2+W(\mathbf x_\perp)+C,\qquad
\hbar\omega=\Delta_\parallel,
\]
where $\mathbf x_\perp:=\mathbf x-(\hat{\mathbf b} \cdot \mathbf x) \hat{\mathbf b}$, $W$ is a confining function of 
the transverse variable, and $C\in\mathbb R$.
\end{theorem}

Equality in \eqref{eq:var-par-bound} is equivalent to $x_\parallel\psi_0$
lying in the active subspace. Hence $(H-E_0)^2 x_\parallel\psi_0=\Delta_\parallel^2 x_\parallel\psi_0$.
On the other hand, \eqref{eq:double-par} gives
$[H,[H,x_\parallel]]\psi_0=(\hbar^2/m) \partial_\parallel V \psi_0$. Using $\psi_0>0$ we get $\partial_\parallel V=m\omega^2 x_\parallel$, with $\omega=\Delta_\parallel/\hbar$.
Integrating in $x_\parallel$ at fixed $\mathbf x_\perp$ yields the required quadratic form along the field; the converse 
follows by direct substitution.

Besides,
\begin{equation}\label{eq:trk-par-const}
[x_\parallel,[H,x_\parallel]]
=- \frac{i\hbar}{m} [x_\parallel, \hat{\mathbf b} \cdot \boldsymbol{\pi}]
=- \frac{i\hbar}{m} i\hbar
=\frac{\hbar^2}{m}.
\end{equation}
Hence the $f$-sum (TRK) for $x_\parallel$ has the same constant as without the field:
\begin{equation}\label{eq:trk-par}
\sum_{n>0}(E_n-E_0) \bigl|\langle\psi_n, x_\parallel \psi_0\rangle\bigr|^2
=\frac12 \big\langle [x_\parallel,[H,x_\parallel]]\big\rangle_0
=\frac{\hbar^2}{2m}.
\end{equation}
This is consistent with the longitudinal form of the harmonic-potential theorem \cite{Kohn1961,Dobson1994HPT}.

Assume \eqref{eq:var-par-bound} holds. In a uniform field one can choose a gauge (e.g., the symmetric gauge) with $A_\parallel=0$,
 so that
\[
H=\Bigl(\frac{p_\parallel^{2}}{2m}+\tfrac12 m\omega^{2}x_\parallel^{2}\Bigr)
 + \Bigl(\frac{\boldsymbol{\pi}_\perp^{2}}{2m}+W(\mathbf x_\perp)\Bigr)+C,
\qquad \hbar\omega=\Delta_\parallel,
\]
i.e., the Hamiltonian splits into a longitudinal harmonic part and a transverse part. Therefore the ground state factorizes
 as $\psi_0(x_\parallel,\mathbf x_\perp)=\phi_0(x_\parallel) \chi_0(\mathbf x_\perp)$, and the matrix elements for 
 $x_\parallel$ decompose:
\[
\langle\psi_n, x_\parallel \psi_0\rangle
=\langle\phi_{n_\parallel}, x_\parallel \phi_0\rangle \cdot \langle\chi_{n_\perp}, \chi_0\rangle.
\]
Thus $x_\parallel$ does not excite the transverse subsystem ($\langle\chi_{n_\perp},\chi_0\rangle=\delta_{n_\perp,0}$), 
and only the longitudinal harmonic mode remains active, where the single nonzero matrix element is 
the transition $0\to1$ with energy $\hbar\omega=\Delta_\parallel$. The TRK along the field is therefore saturated by one
 term, and the variance equals
\[
\Var_0(x_\parallel)=\frac{\hbar}{2m\omega}=\frac{\hbar^2}{2m \Delta_\parallel},
\]
which gives equality in \eqref{eq:var-par-bound}. The theorem is proved. \hfill$\square$

\begin{lemma}[Equality equivalences for $x_\parallel$]\label{lem:eq-par}
The following statements are equivalent:
\begin{enumerate}\itemsep0.25em
\item $\Var_0(x_\parallel)=\hbar^2/(2m\Delta_\parallel)$;
\item $x_\parallel\psi_0\in\mathrm{span}\{\psi^\parallel_1\}$, i.e., $x_\parallel$ excites a single longitudinal mode with 
gap $\Delta_\parallel$;
\item the TRK sum \eqref{eq:trk-par} is saturated by the single transition $0 \to 1_\parallel$;
\item $[H,[H,x_\parallel]] \psi_0=\Delta_\parallel^2 x_\parallel\psi_0$, i.e., $\partial_\parallel V$ is linear in 
$x_\parallel$ on the support of $\rho_0$.
\end{enumerate}
\end{lemma}

\begin{remark}[On the choice of gap]
Using $\Delta_\parallel$ from \eqref{eq:gap-par} is essential: the global first gap $E_1-E_0$ may belong to a transverse 
(cyclotron) mode that $x_\parallel$ does not excite; then $\Delta_\parallel>E_1-E_0$ and \eqref{eq:var-par-bound} becomes 
sharper. In any case, the bound with $E_1-E_0$ remains valid but is not always optimal.
\end{remark}

\begin{remark}[TRK constant is unchanged]
Formula \eqref{eq:trk-par-const} shows that for $x_\parallel$ the $f$-sum has the same constant $\hbar^2/(2m)$ as at $B=0$. 
This reflects the fact that a magnetic field does no work along its own direction and does not affect the longitudinal 
''dipole'' $f$-sum. See also the generalized HPT formula for time-dependent fields \cite{LaiPan2016SciRep}.
\end{remark}

\subsection{Across the field: guiding center $R_\perp$}
\label{sec:magnetic-transverse}

Let $\mathbf B=B \hat{\mathbf b}$ be uniform, with $q\neq0$. Introduce the kinetic momentum 
$\boldsymbol{\pi}=\mathbf p-q \mathbf A(\mathbf x)$, so that $[\pi_i,\pi_j]=i\hbar q \varepsilon_{ijk}B_k$, 
and the magnetic length $\ell_B^2:=\hbar/(|q|B)$. Let $P_\perp:=I-\hat{\mathbf b}\hat{\mathbf b}^{ \top}$ be the projector 
onto the plane perpendicular to the field. 

Define the guiding center (in the direction transverse to the field) by
\begin{equation}\label{eq:Rperp-def}
\mathbf R_\perp := P_\perp \left(\mathbf x+\frac{1}{qB} \hat{\mathbf b}\times\boldsymbol{\pi}\right)
 = P_\perp \left(\mathbf x+\mathrm{sgn}(q) \frac{\ell_B^2}{\hbar} \hat{\mathbf b}\times\boldsymbol{\pi}\right).
\end{equation}
For any unit $u\perp\hat{\mathbf b}$ set $R_u:=u \cdot \mathbf R_\perp$, and $w:=u\times\hat{\mathbf b}$ 
(this is also a unit vector in $P_\perp$).

Direct calculations for a uniform field give the basic relations

\begin{equation}\label{eq:comm-basic-R}
[\pi_i, R_{\perp,j}]=0,\qquad
[R_{\perp,i}, R_{\perp,j}]
= i \mathrm{sgn}(q) \ell_B^2 \varepsilon_{ijk} \hat b_k,\qquad
[H_0:=\boldsymbol{\pi}^2/2m, \mathbf R_\perp]=0.
\end{equation}

For any smooth $f(\mathbf x)$ and any $u\perp\hat{\mathbf b}$ we have

\begin{equation}\label{eq:R-f-comm}
[R_u, f(\mathbf x)]
=\mathrm{sgn}(q) \frac{\ell_B^2}{\hbar} [ u \cdot(\hat{\mathbf b}\times\boldsymbol{\pi}), f ]
=- i \mathrm{sgn}(q) \ell_B^2 w \cdot\nabla f,
\end{equation}
and in particular,

\begin{equation}\label{eq:H-Ru}
[H, R_u]=[V, R_u]=i \mathrm{sgn}(q) \ell_B^2 w \cdot\nabla V(\mathbf x).
\end{equation}

Applying \eqref{eq:R-f-comm} to $f=w \cdot\nabla V$ we get

\begin{equation}\label{eq:double-Ru}
\boxed{\quad
[R_u, [H, R_u]]
=i \mathrm{sgn}(q) \ell_B^2 [R_u, w \cdot\nabla V]
=\ell_B^4 (w^{ \top}(\nabla\nabla V)w),
\quad}
\end{equation}
i.e., a pure multiplier equal to the projection of the Hessian of $V$ onto $w$. Therefore $R_u$ satisfies the exact $f$-sum rule

\begin{equation}\label{eq:trk-Ru}
\sum_{n>0}(E_n-E_0) \bigl|\langle \psi_n, R_u \psi_0\rangle\bigr|^2
=\frac12 \big\langle [R_u,[H,R_u]]\big\rangle_0
=\frac{\ell_B^4}{2} \Big\langle w^{ \top}(\nabla\nabla V)w \Big\rangle_0
=:S_{R_u}.
\end{equation}

Let us define
\begin{equation}\label{eq:Delta-Ru}
\Delta_{R_u}:=\min\{ E_n-E_0: \langle \psi_n, R_u \psi_0\rangle\neq0 \}.
\end{equation}

Then, combining the ''Chebyshev step'' with the active gap and the TRK formula \eqref{eq:trk-Ru}, we obtain
\begin{equation}\label{eq:var-Ru-bound}
\boxed{\qquad
\Var_0(R_u) \le \frac{1}{\Delta_{R_u}}\sum_{n>0}(E_n-E_0) \bigl|\langle \psi_n, R_u \psi_0\rangle\bigr|^2
=\frac{\ell_B^4}{2 \Delta_{R_u}} 
\Big\langle w^{ \top}(\nabla\nabla V)w \Big\rangle_0.
\qquad}
\end{equation}

{{Importantly, the upper bound \eqref{eq:var-Ru-bound} itself does not require any pointwise
(constant-curvature) assumption.
Indeed, using the exact identity \eqref{eq:double-Ru} and taking the ground-state expectation, one has
\[
S_{R_u}=\tfrac12\langle[ R_u,[H,R_u] ]\rangle_0
=\frac{\ell_B^{4}}{2} \big\langle w^{\top}(\nabla\nabla V)w\big\rangle_0.
\]
So \eqref{eq:var-Ru-bound} is controlled solely by the ground-state averaged projected curvature.
The corresponding quantitative rigidity estimates \eqref{eq:D1-Ru}-\eqref{eq:D2-Ru}
are governed by the same averaged quantity.}}

Equality in \eqref{eq:var-Ru-bound} is equivalent to $R_u\psi_0$ lying in the active subspace,
i.e., $R_u\psi_0\in\mathrm{span}\{\psi_{1,u}\}$ where $E_{1,u}-E_0=\Delta_{R_u}$ and a single mode is excited. In this case
\[
2\Delta_{R_u} \Var_0(R_u)=\big\langle [R_u,[H,R_u]]\big\rangle_0,
\]
and from \eqref{eq:double-Ru} we get the necessary condition that $w^{ \top}(\nabla\nabla V)w$ must be 
constant on average (and if it is in fact constant on the support of $\rho_0$, then it is a constant as a function). 
In particular, if equality holds for two linearly independent directions $u_1,u_2\perp\hat{\mathbf b}$ 
(and thus for two independent $w_k=u_k\times\hat{\mathbf b}$), then the curvature of $V$ is constant throughout the plane
 $P_\perp$, and therefore
\begin{equation}\label{eq:V-quadr-Perp}
V(\mathbf x)=\frac{m}{2}\big(\Omega_x^2 x_\perp^2+\Omega_y^2 y_\perp^2\big)
+W(x_\parallel)+C,
\end{equation}
i.e., $V$ is quadratic on $P_\perp$ (with arbitrary dependence along the field).

\begin{remark}
Unlike the longitudinal case (§\ref{sec:magnetic-parallel}), the operator $[H,[H,R_u]]$ is not a multiplier 
(it contains $\boldsymbol{\pi}$), so an exact ''iff'' statement in the transverse direction requires a separate analysis.
 Formula \eqref{eq:double-Ru} is useful because it yields a multiplier precisely in the TRK commutator $[R_u,[H,R_u]]$, 
 which is sufficient for the sharp bound and for the quantitative rigidity below.
\end{remark}

Let \(\mathbf B=B\hat{\mathbf b}\) be uniform, \(q\neq0\); fix unit vectors \(u\perp\hat{\mathbf b}\) and \(w:=u\times\hat{\mathbf b}\).

{{We stress that the upper bound \eqref{eq:var-Ru-bound} and the rigidity estimates
\eqref{eq:D1-Ru}-\eqref{eq:D2-Ru} require only the ground-state averaged curvature
$\langle w^{\top}(\nabla\nabla V)w\rangle_0$.
The additional pointwise structural assumptions introduced below are used
solely to convert the spectral saturation condition
($R_u\psi_0$ lying entirely in the active subspace)
into an explicit necessary and sufficient statement on the potential $V$.}}

Assume on the support of \(\rho_0\) the two structural hypotheses hold: (i) constant transverse curvature along \(w\),
\[
w^\top(\nabla\nabla V)w \equiv m\Omega^2 \quad(\text{constant}),
\]
and (ii) no transverse ''mixing'' along \(w\): \(\partial_u\partial_w V\equiv0\) and \(\partial_\parallel\partial_w V\equiv0\). Equivalently,
\[
V(\mathbf x)=\frac{m\Omega^2}{2} \bigl(w \cdot \mathbf x_\perp\bigr)^2+W \bigl(u \cdot \mathbf x_\perp, x_\parallel\bigr),
\qquad \mathbf x_\perp:=\mathbf x-(\hat{\mathbf b} \cdot \mathbf x) \hat{\mathbf b}.
\]
Then for the guiding-center component \(R_u:=u \cdot \mathbf R_\perp\) we have the sharp bound
\[
\Var_0(R_u) \le \frac{\ell_B^4}{2 \Delta_{R_u}} 
\big\langle w^\top(\nabla\nabla V)w\big\rangle_0
=\frac{\ell_B^4 m\Omega^2}{2 \Delta_{R_u}},
\]
and equality holds \emph{if and only if} \(\Delta_{R_u}=\hbar\Omega\) (that is, the pair \((R_u, u \cdot \boldsymbol{\pi})\) realizes a 
single harmonic mode) and \(R_u\psi_0\in\mathrm{span}\{\psi_{1,u}\}\).

Moreover, if equality holds for \emph{two linearly independent} directions \(u_1,u_2\perp\hat{\mathbf b}\), then \(V\) is quadratic on the 
whole transverse plane \(P_\perp\):
\[
V(\mathbf x)=\frac{m}{2}\bigl(\Omega_x^2 x_{\perp,x}^2+\Omega_y^2 x_{\perp,y}^2\bigr)+W(x_\parallel)+C,
\]
and in that case saturation holds for every \(u\perp\hat{\mathbf b}\).

\medskip
\begin{proof}
From \eqref{eq:double-Ru} we have the exact TRK double-commutator
\(
[R_u,[H,R_u]]=\ell_B^4  w^{ \top}(\nabla\nabla V)w.
\)
Under (i) it equals the constant \(\ell_B^4 m\Omega^2\), hence
\[
S_{R_u}=\tfrac12\langle [R_u,[H,R_u]]\rangle_0=\frac{\ell_B^4 m\Omega^2}{2}.
\]
Therefore, the master bound for \(A=R_u\) (see \eqref{eq:masterA}) gives
\(
\Var_0(R_u)\le S_{R_u}/\Delta_{R_u}=\ell_B^4 m\Omega^2/(2\Delta_{R_u}),
\)
which is the desired upper bound.

Next we use the equality criteria from §\ref{sec:setup}: in \eqref{eq:masterA} equality is equivalent to
\(
R_u\psi_0\in \mathrm{Ran} P_{E_0+\Delta_{R_u}}
\)
(the entire weight of \(R_u\psi_0\) lies in the active subspace) and, consequently,
\(
(H-E_0) R_u\psi_0=\Delta_{R_u} R_u\psi_0.
\)
On the other hand,
\(
[H,R_u]=i \mathrm{sgn}(q) \ell_B^2  w \cdot\nabla V
\)
(see \eqref{eq:H-Ru}). Assumptions (ii) and (i) imply
\(
w \cdot\nabla V = \partial_{w}V = m\Omega^2 (w \cdot\mathbf x_\perp),
\)
and therefore
\begin{equation}\label{eq:comm-eig}
(H-E_0) R_u\psi_0  =  [H,R_u]\psi_0  =  i \mathrm{sgn}(q) \ell_B^2  m\Omega^2 (w \cdot\mathbf x_\perp) \psi_0 .
\end{equation}
\end{proof}

Applying \((H-E_0)\) to \eqref{eq:comm-eig} once more and using \([R_u,[H,R_u]]=\ell_B^4 m\Omega^2\), we obtain
\[
(H-E_0)^2 R_u\psi_0
= \ell_B^4 m\Omega^2  R_u\psi_0 ,
\]
i.e., \(R_u\psi_0\) is an eigenvector of \((H-E_0)^2\) with eigenvalue \(\ell_B^4 m\Omega^2\).
Comparing with the equality case in \eqref{eq:masterA} (where \((H-E_0)R_u\psi_0=\Delta_{R_u} R_u\psi_0\)), we conclude that
\(
\Delta_{R_u}^2=\hbar^2\Omega^2
\)
and, since \(\Delta_{R_u}>0\), \(\Delta_{R_u}=\hbar\Omega\).
Thus, equality in \eqref{eq:var-Ru-bound} holds \(\Longleftrightarrow\)
\(\Delta_{R_u}=\hbar\Omega\) and \(R_u\psi_0\) is entirely supported in the active subspace, as required.

If equality is attained for two linearly independent \(u_1,u_2\perp\hat{\mathbf b}\), then by the same steps
\(w_k^{ \top}(\nabla\nabla V)w_k\equiv \text{const}\) for \(k=1,2\), and the mixed derivatives in assumption (ii) vanish. This forces the 
Hessian of \(V\) 
to be constant on the whole plane \(P_\perp\). Hence \(V\) is quadratic on \(P_\perp\) and the bound is saturated for any \(u\perp\hat{\mathbf b}\).

\paragraph{Quantitative rigidity (transverse).}
Define the deficit \(\varepsilon_u:=\dfrac{S_{R_u}}{\Delta_{R_u}}-\Var_0(R_u)\ge0\), the second gap
\[
\Gamma_{R_u}:=\min\{ E_n-E_{1,u}: \langle\psi_n, R_u \psi_0\rangle\neq0, n\ge2 \},
\]
and the TRK tail fraction for \(R_u\):
\begin{equation}\label{eq:eta-Ru}
\eta_{\mathrm{TRK}}^{(R_u)}
:=1-\frac{\Delta_{R_u} |\langle \psi_{1,u}, R_u \psi_0\rangle|^2}{S_{R_u}}\in[0,1].
\end{equation}

Then purely spectral algebra (exactly as in D1-D2) yields explicit constants:
\begin{equation}\label{eq:D1-Ru}
\boxed{\quad
\varepsilon_u \ge \frac{\Gamma_{R_u}}{\Delta_{R_u}(\Delta_{R_u}+\Gamma_{R_u})} S_{R_u} 
\eta_{\mathrm{TRK}}^{(R_u)}
 = \frac{\ell_B^4}{2} 
\frac{\Gamma_{R_u}}{\Delta_{R_u}(\Delta_{R_u}+\Gamma_{R_u})} 
\Big\langle w^{ \top}(\nabla\nabla V)w \Big\rangle_0 
\eta_{\mathrm{TRK}}^{(R_u)}  ,
\quad}
\end{equation}
\begin{equation}\label{eq:D2-Ru}
\boxed{\quad
\sum_{n\ge2}\bigl|\langle \psi_n, R_u \psi_0\rangle\bigr|^2 \le \frac{\Delta_{R_u}}{\Gamma_{R_u}} \varepsilon_u,
\qquad
\eta_{\mathrm{TRK}}^{(R_u)} \le \frac{\Delta_{R_u}(\Delta_{R_u}+\Gamma_{R_u})}{S_{R_u} \Gamma_{R_u}} \varepsilon_u.
\quad}
\end{equation}

In other words, a small deficit quantitatively suppresses both the \(R_u\)-variance tail and the TRK tail for \(R_u\).

\paragraph{Remarks.}
(i) In an isotropic transverse quadratic trap one has \(w^{ \top}(\nabla\nabla V)w\equiv\text{const}\), and the bound \eqref{eq:var-Ru-bound} is sharp. 
(ii) All formulas are consistent with the scaling \(\ell_B^2\propto B^{-1}\). They are also invariant under the replacements \(u\mapsto u'\), \(w=u\times\hat{\mathbf b}\). 
(iii) The longitudinal results of §\ref{sec:magnetic-parallel} are obtained as the degenerate case \(w=0\).

\subsection{Inhomogeneous magnetic field with fixed direction: full ''iff'' along the field}
\label{sec:magnetic-parallel-inhom}

Let $\mathbf B(\mathbf x)=\beta(\mathbf x) \hat{\mathbf b}$ have a fixed direction $\hat{\mathbf b}$ and smooth scalar strength $\beta(\mathbf x)$, and let
 $q\neq0$. Consider
\[
H=\frac{\boldsymbol{\pi}^2}{2m}+V(\mathbf x),\qquad \boldsymbol{\pi}=\mathbf p-q\mathbf A,\qquad \nabla\times \mathbf A=\mathbf B,
\]
and denote $x_\parallel:=\hat{\mathbf b} \cdot \mathbf x$, $\partial_\parallel:=\hat{\mathbf b} \cdot\nabla$, as well as
\[
\Delta_\parallel:=\min\{ E_n-E_0: \langle \psi_n, x_\parallel \psi_0\rangle\neq0 \}.
\]

\begin{theorem}[Sharp longitudinal bound and \emph{iff} for $\mathbf B(\mathbf x)\parallel\hat{\mathbf b}$]\label{thm:iff-par-inhomB}
Always
\begin{equation}\label{eq:var-par-inhomB}
\Var_0(x_\parallel) \le \frac{\hbar^2}{2m \Delta_\parallel}.
\end{equation}
Equality in \eqref{eq:var-par-inhomB} holds \emph{if and only if}
\begin{equation}\label{eq:V-parabolic-inhomB}
V(\mathbf x)=\frac{m\omega^2}{2} x_\parallel^2+W(\mathbf x_\perp)+C,\qquad \hbar\omega=\Delta_\parallel,
\end{equation}
where $\mathbf x_\perp:=\mathbf x-(\hat{\mathbf b} \cdot \mathbf x)\hat{\mathbf b}$, $W$ is any confining function of $\mathbf x_\perp$ only, and $C\in\mathbb R$.
\end{theorem}

\begin{proof}
As in the homogeneous case, $[H,x_\parallel]=-\frac{i\hbar}{m} \hat{\mathbf b} \cdot \boldsymbol{\pi}$, hence
\[
[H,[H,x_\parallel]]
=-\frac{i\hbar}{m} [H, \hat{\mathbf b} \cdot \boldsymbol{\pi}]
=-\frac{i\hbar}{m}\Bigl(\frac{1}{2m} [\boldsymbol{\pi}^2, \hat{\mathbf b} \cdot \boldsymbol{\pi}]+[V, \hat{\mathbf b} \cdot \boldsymbol{\pi}]\Bigr).
\]
Since $\mathbf B(\mathbf x)=\beta(\mathbf x)\hat{\mathbf b}$, we have
\[
[\pi_i, \hat{\mathbf b} \cdot \boldsymbol{\pi}]=i\hbar q \varepsilon_{ijk}\hat b_j B_k(\mathbf x)
=i\hbar q \beta(\mathbf x) \varepsilon_{ijk}\hat b_j\hat b_k=0,
\]
and therefore $[\boldsymbol{\pi}^2, \hat{\mathbf b} \cdot \boldsymbol{\pi}]=0$. This does not require $\beta$ to be homogeneous. Moreover, 
$[V, \hat{\mathbf b} \cdot \boldsymbol{\pi}]=i\hbar \partial_\parallel V$, so
\[
[H,[H,x_\parallel]]=\frac{\hbar^2}{m} \partial_\parallel V(\mathbf x).
\]
The TRK double-commutator for $x_\parallel$ equals $[x_\parallel,[H,x_\parallel]]=\hbar^2/m$, hence
\(
\Var_0(x_\parallel)\le \hbar^2/(2m\Delta_\parallel),
\)
which gives \eqref{eq:var-par-inhomB}. If equality holds, then as before
\(
(H-E_0)^2 x_\parallel\psi_0=\Delta_\parallel^2 x_\parallel\psi_0,
\)
and from the last formula it follows that
\(
\partial_\parallel V=m\omega^2 x_\parallel
\)
with $\omega=\Delta_\parallel/\hbar$. Integrating in $x_\parallel$ at fixed $\mathbf x_\perp$ yields \eqref{eq:V-parabolic-inhomB}; the converse implication 
is verified by direct substitution.
\end{proof}

\begin{remark}[A priori non-separability]
No separability of $V$ is assumed in the statement. The equality condition itself forces the form \eqref{eq:V-parabolic-inhomB}: any mixed terms of 
the type $x_\parallel \Phi(\mathbf x_\perp)$ contradict the linearity $\partial_\parallel V=m\omega^2 x_\parallel$.
\end{remark}

\section{Consequences and applications}\label{sec:conseq}

This section collects immediate physical consequences of the sharp bounds and quantitative rigidity from 
§§\ref{sec:theoremA}-\ref{sec:rigidity}, as well 
as their magnetic counterparts from §§\ref{sec:magnetic-parallel}, \ref{sec:magnetic-transverse}, 
\ref{sec:magnetic-parallel-inhom}.

\subsection{Static dipole polarizability}

The static polarizability along \(x\) is given by the Kramers-Heisenberg sum (see, e.g., \cite{Summa2023JCTC,Szabo2022PRL,Goger2024JCTC,Cheng2024PRA})
\[
\alpha(0) = 2\sum_{n>0}\frac{|x_{n0}|^2}{E_n-E_0}
 = \frac{2}{\Delta}\sum_{n>0}\frac{\Delta}{E_n-E_0} |x_{n0}|^2,
\]
where \(x_{n0}:=\langle\psi_n, x \psi_0\rangle\). Since \(E_n-E_0\ge\Delta\), we obtain
\begin{equation}\label{eq:alpha-upper-1}
\alpha(0) \le \frac{2}{\Delta}\sum_{n>0}|x_{n0}|^2
 = \frac{2 \Var_0(x)}{\Delta}.
\end{equation}
Combining \eqref{eq:alpha-upper-1} with the sharp bound of \eqref{eq:mainbound}, for a fixed gap \(\Delta\) we have
\begin{equation}\label{eq:alpha-upper-2}
\boxed{\qquad
\alpha(0) \le \frac{\hbar^2}{m \Delta^2},
\qquad
\text{and equality holds \emph{iff} } V(x)=\tfrac12 m\omega^2x^2+C, \hbar\omega=\Delta.
\qquad}
\end{equation}
Physically, it means that at a given operating frequency (energy gap), no anharmonic confinement can increase the static polarizability beyond that of the harmonic trap.

For a homogeneous field $\mathbf B=B\hat{\mathbf b}$ and the longitudinal coordinate $x_\parallel=\hat{\mathbf b} \cdot \mathbf x$, we similarly have
\begin{equation}\label{eq:alpha-par}
\alpha_\parallel(0) = 2\sum_{n>0}\frac{|\langle\psi_n, x_\parallel \psi_0\rangle|^2}{E_n-E_0}
 \le \frac{2 \Var_0(x_\parallel)}{\Delta_\parallel}
 \le \frac{\hbar^2}{m \Delta_\parallel^{ 2}},
\end{equation}
with equality \emph{iff} $V$ is quadratic along $\hat{\mathbf b}$ (Theorem~\ref{thm:iff-par}). 

In the transverse plane it is natural to consider the response of
 the guiding-center coordinate $R_u$.
  Defining the ''transverse polarizability''
\[
\alpha_{R_u}(0):=2\sum_{n>0}\frac{|\langle\psi_n, R_u \psi_0\rangle|^2}{E_n-E_0},
\]
we obtain from \eqref{eq:var-Ru-bound} 
\begin{equation}\label{eq:alpha-Ru}
\boxed{\qquad
\alpha_{R_u}(0) \le \frac{2 \Var_0(R_u)}{\Delta_{R_u}}
 \le \frac{\ell_B^4}{\Delta_{R_u}^{ 2}} 
\Big\langle w^{ \top}(\nabla\nabla V)w \Big\rangle_0,
\qquad w=u\times\hat{\mathbf b}.
\qquad}
\end{equation}
Equality in \eqref{eq:alpha-Ru} is possible only when the transverse curvature along $w$ is constant
 (see §\ref{sec:magnetic-transverse}).

\subsection{Quantum metric (fidelity susceptibility)}

For the unitary family generated by the position operator, \(U(\lambda)=e^{-i\lambda x/\hbar}\), the quantum metric along \(x\) is
\begin{equation}\label{eq:gxx-def}
g_{xx} = \frac{\Var_0(x)}{\hbar^2}.
\end{equation}
From \eqref{eq:mainbound} we immediately obtain the sharp bound
\begin{equation}\label{eq:gxx-bound}
\boxed{\qquad
g_{xx} \le \frac{1}{2m \Delta} ,
\qquad
\text{and equality holds \emph{iff} } V(x)=\tfrac12 m\omega^2x^2+C, \hbar\omega=\Delta.
\qquad}
\end{equation}
In a magnetic field, for parametric shifts of the guiding center
\(\mathbf R_\perp\mapsto \mathbf R_\perp+\lambda u\) we have
\begin{equation}\label{eq:gRu-bound}
g_{R_u R_u} = \frac{\Var_0(R_u)}{\hbar^2}
 \le \frac{\ell_B^4}{2\hbar^2 \Delta_{R_u}} 
\Big\langle w^{ \top}(\nabla\nabla V)w \Big\rangle_0,
\qquad w=u\times\hat{\mathbf b}.
\end{equation}
Saturation is possible only in case of constant curvature along $w$ (see §\ref{sec:magnetic-transverse}).

Theorems D1-D3 directly control the ''metric deficit''
\[
\delta g_{xx}:=\frac{1}{2m\Delta}-g_{xx}
=\frac{\varepsilon}{\hbar^2},
\]
in terms of \(\Gamma\), \(\eta_{\mathrm{TRK}}\), and the anharmonicity measure
\( V'-m\omega^2x\)  in \({L^2(\rho_0)}\). Analogous bounds hold for \(g_{R_u R_u}\) using \eqref{eq:D1-Ru}-\eqref{eq:D2-Ru}.

For a parametric shift generated by a self-adjoint \(A\), the quantum metric
\(g_{AA}=\Var_0(A)/\hbar^2\) satisfies
\begin{equation}\label{eq:metricA}
g_{AA} = \frac{\Var_0(A)}{\hbar^2}
 \le  \frac{1}{2\hbar^2\Delta_A} \big\langle [A,[H,A]]\big\rangle_0.
\end{equation}
The equality criterion is the same as above: \(A\psi_0\) is entirely supported in
the active subspace with gap \(\Delta_A\).
The special cases \(A=x\), \(A=x_\parallel\), \(A=R_u\), and \(A=p\) follow from \eqref{eq:metricA} by substitution.

\subsection{A corridor for the momentum variance}

For $p=-i\hbar \partial_x$ we have
\(
[H,p]=[V,p]=i\hbar V'(x)
\)
and
\(
[p,[H,p]]=\hbar^2 V''(x)
\).
Therefore the $f$-sum rule gives

\begin{equation}\label{eq:trk-p}
\sum_{n>0}(E_n-E_0) |p_{n0}|^2
=\frac12 \big\langle [p,[H,p]]\big\rangle_0
=\frac{\hbar^2}{2} \langle V''\rangle_0,
\end{equation}
where $p_{n0}:=\langle\psi_n, p \psi_0\rangle$. From this, the standard bound by the active gap
yields an \emph{upper} estimate
\begin{equation}\label{eq:Varp-upper}
\boxed{\qquad
\Var_0(p)=\sum_{n>0}|p_{n0}|^2
 \le \frac{1}{\Delta}\sum_{n>0}(E_n-E_0)|p_{n0}|^2
=\frac{\hbar^2}{2\Delta} \langle V''\rangle_0.
\qquad}
\end{equation}

Assume, in addition, that the second gap $\Gamma=E_2-E_1$ is known and that for all $n\ge2$
the contributions $|p_{n0}|^2$ can occur only when $E_n-E_0\ge \Delta+\Gamma$. Then from \eqref{eq:trk-p}
we also obtain a lower estimate

\begin{equation}\label{eq:Varp-lower}
\boxed{\qquad
\Var_0(p)
 \ge \frac{1}{\Delta+\Gamma}\sum_{n>0}(E_n-E_0)|p_{n0}|^2
=\frac{\hbar^2}{2(\Delta+\Gamma)} \langle V''\rangle_0.
\qquad}
\end{equation}
Together, \eqref{eq:Varp-upper}-\eqref{eq:Varp-lower} provide a corridor for $\Var_0(p)$ at fixed $\Delta,\Gamma$ and mean curvature $\langle V''\rangle_0$.

The uncertainty principle gives \(\Var_0(x)\Var_0(p)\ge\hbar^2/4\). Combining this with \eqref{eq:mainbound},
we obtain the universal lower estimate

\begin{equation}\label{eq:Varp-floor}
\boxed{\qquad \Var_0(p) \ge \frac{m \Delta}{2} , \qquad}
\end{equation}
independent of the shape of the confining potential \(V\).
In a homogeneous magnetic field, for the longitudinal component
\(p_\parallel=\hat{\mathbf b} \cdot \mathbf p\) we similarly have
\begin{equation}\label{eq:Varp-par-floor}
\Var_0(p_\parallel) \ge \frac{m \Delta_\parallel}{2} .
\end{equation}

\subsection{Numerical illustration (anharmonic oscillator)}
Consider the family \(V(x)=\tfrac12 m\omega^2 x^2+\lambda x^4\) in units \(\hbar=m=\omega=1\).
For the ground state \(\psi_0\) and first excited state \(\psi_1\) we compute $E_0$,$E_1$, and the gap \(\Delta=E_1-E_0\). 
We also evaluate the variance \(\Var_0(x)\) and the deficit
\[
\varepsilon  :=  \frac{\hbar^2}{2m \Delta}-\Var_0(x) \ge 0,
\]
as well as the structural norm of the deviation from harmonicity
\[
D_3:=\bigl(  V'(x)-m\omega^2 x \bigr) _{L^2(\rho_0)},\qquad \rho_0=|\psi_0|^2.
\]
The exact identity \eqref{eq:eps-exact} expresses \(\varepsilon\) in terms of spectral contributions.

\begin{figure}[t]
  \centering
  \includegraphics[width=\linewidth]{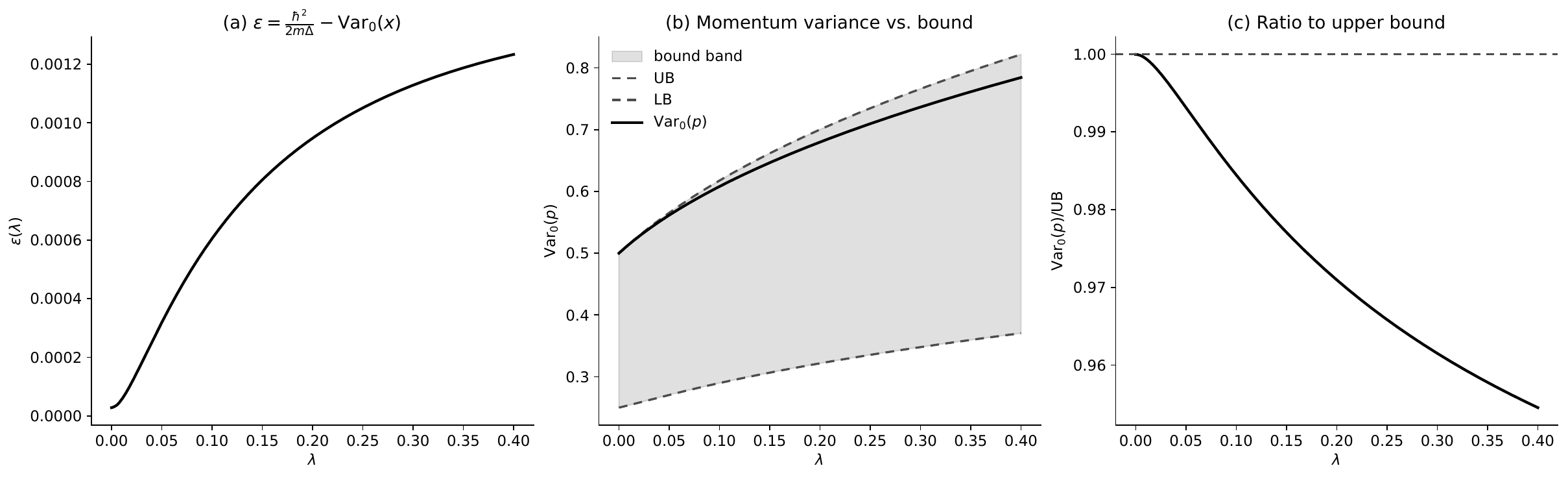}
  \caption{Numerical illustration for the anharmonic oscillator $V(x)=\tfrac12 x^2+\lambda x^4$
(units $\hbar=m=\omega=1$). Horizontal axis: $\lambda$.
(a) Deficit $\varepsilon(\lambda)$ from the sharp bound on $\Var_0(x)$.
(b) Ground-state momentum variance $\Var_0(p)$ compared with the theoretical bounds from Sec.~6.3.
(c) Ratio of $\Var_0(p)$ to the corresponding upper bound.}

  \label{fig:varp-corridor}
\end{figure}

\paragraph{Why these plots suffice.}
Panel (a) shows that the deficit
\(\varepsilon(\lambda)=\tfrac{\hbar^{2}}{2m\Delta}-\mathrm{Var}_0(x)\) 
is non-negative for the anharmonic family \(V(x)=\tfrac12 x^{2}+\lambda x^{4}\) and vanishes at \(\lambda=0\). This directly 
tests the sharp bound \(\mathrm{Var}_0(x)\le \hbar^{2}/(2m\Delta)\) and its saturation in the harmonic limit. 

Panels (b)-(c) 
validate the momentum-variance ''corridor'' from Sec.~6.3. The numerical values of \(\mathrm{Var}_0(p)\) remain inside the band 
\([\mathrm{LB},\mathrm{UB}]\), with 
\(\mathrm{UB}= \tfrac{\hbar^{2}}{2\Delta}\langle V''\rangle_0\) 
and 
\(\mathrm{LB}= \tfrac{\hbar^{2}}{2(\Delta+\Gamma)}\langle V''\rangle_0\). 
The ratio plot \(\mathrm{Var}_0(p)/\mathrm{UB}\le 1\) highlights saturation at \(\lambda=0\) and strict inequality for \(\lambda>0\). 

Together, (a)-(c) confirm our two core claims: (i) the extremal bound at fixed active gap \(\Delta\), and (ii) quantitative rigidity away 
from the harmonic point. The monotone growth of \(\varepsilon\) with \(\lambda\) signals departure from single-transition dominance 
and, by D1-D2, the associated growth of the TRK tail (not plotted here but implied by the theory).

\subsection{Practical implications}

\begin{itemize}
\item For a given operating gap \(\Delta\),
the bounds \eqref{eq:alpha-upper-2}, \eqref{eq:gxx-bound}, and \eqref{eq:Varp-upper}
set universal limits on the polarizability, the ''geometric'' sensitivity,
and the ground-state momentum spread; equality is achieved only by the harmonic trap.
\item  The deficit \(\varepsilon\) is quantitatively linked
to the TRK tail and the norm \( V'-m\omega^2x _{L^2(\rho_0)}\) (see Theorems~\ref{thm:D1}-\ref{thm:D3});
the same language carries over to the transverse direction for \(R_u\) (§\ref{sec:magnetic-transverse}).
\item  For ions/cold atoms/quantum dots,
\(\Delta\) and \(\Gamma\) are extracted spectroscopically, while \(\Var_0(x)\) can be obtained
from density images or via the metric \(g_{xx}\).
Testing \eqref{eq:alpha-upper-2}-\eqref{eq:gxx-bound} provides a strict sanity check
and serves as a diagnostic of proximity to the harmonic regime.
\end{itemize}

\begin{remark}[Why there is no general transverse iff for inhomogeneous \(B\)]
If \(u\perp\hat{\mathbf b}\) and \(\nabla_\perp B\neq0\), then
\([ \pi^2, u \cdot \boldsymbol{\pi} ]\neq0\). The commutator contains both multiplicative terms from \(\mathbf B(\mathbf x)\)
and a symmetric part involving \(\boldsymbol{\pi}\).
It follows that the operator \([H,[H,R_u]]\) ceases to be a pure multiplier.
Therefore, without additional structural hypotheses (e.g., \(B(\mathbf x)=\beta(x_\parallel)\hat{\mathbf b}\) and
\(w^{ \top}(\nabla\nabla V)w\equiv\text{const}\)) a general iff is unattainable.
Within such a class, however, our results remain sharp as upper bounds and yield quantitative rigidity via 
\eqref{eq:D1-Ru}-\eqref{eq:D2-Ru}.
\end{remark}

\section{Conclusion}\label{sec:conclusion}

We solve the static isoperimetric problem for the Mandelstam-Tamm bound in one-dimensional quantum mechanics. For a fixed 
spectral gap \(\Delta=E_{1}-E_{0}\) we establish the \emph{sharp} upper bound
\[
\Var_{0}(x) \le \frac{\hbar^{2}}{2m \Delta},
\]
and prove \emph{iff}-rigidity: equality holds if and only if the potential is harmonic along the corresponding direction 
(see \Cref{eq:mainbound}). This yields a rigorous static analogue of the quantum speed limit in geometric terms. 
For the Fubini-Study metric \(g_{xx}=\Var_{0}(x)/\hbar^{2}\) one has \(g_{xx}\le (2m\Delta)^{-1}\) with the same equality 
criterion (see \Cref{sec:conseq}).

A key feature is \emph{quantitative rigidity}. 
The deficit
\(
\varepsilon=\frac{\hbar^{2}}{2m\Delta}-\Var_{0}(x)\ge0
\)
is controlled by (i) the TRK tail via the second gap \(\Gamma\) (see \Cref{thm:D1,thm:D2}) and (ii) the 
structural deviation \(V'(x)-m\omega^{2}x\) in the \(L^{2}(\rho_{0})\) norm with explicit constants (\Cref{thm:D3}). 
These estimates turn ''near saturation'' into a measurable scale of trap anharmonicity and yield sharp bounds for the 
polarizability and the quantum metric (\Cref{sec:conseq}).

We extend the picture to magnetic systems.
 Along the field direction (including inhomogeneous fields of fixed direction)
 we obtain the same sharp bound and a full \emph{iff} at fixed active gap \(\Delta_{\parallel}\) 
 (see \Cref{thm:iff-par,thm:iff-par-inhomB}). 
 
In the transverse plane, for the guiding-center coordinate \(R_{u}\), 
 we establish an exact TRK formula through the projected Hessian of \(V\) and the corresponding rigidity bounds 
 (see \Cref{sec:magnetic-transverse}). A global \emph{iff} is achieved within a natural structural 
 class of potentials with constant transverse curvature.

The bounds provide universal limits on sensitivity (metric, polarizability) and on the ground-state momentum spread at a 
given operating gap—independent of the specific confining potential. The quantitative rigidity estimates allow one to calibrate 
anharmonicity from spectroscopic data (\(\Delta,\Gamma\)) and from the ground-state density \(\rho_{0}\), which is directly 
useful for the design and diagnostics of ion/atom and optomechanical traps.

The open questions are:

(1) In the structural estimate of \(\varepsilon\) via \( V'-m\omega^{2}x _{L^{2}(\rho_{0})}\) the constants are optimal 
in order, but their precise optimality across broad classes of \(V\) remains open. 

(2) In the transverse direction with an inhomogeneous field a general \emph{iff} is unattainable without extra assumptions; 
it is of interest to identify a minimal set of hypotheses that guarantees saturation. 

(3) Multidimensional versions with several active channels and possible ground-state degeneracies require a separate analysis 
of domains and symmetry selection. 

(4) It is promising to extend the results to many-body Hamiltonians (effective masses, interactions) and to dynamical scenarios 
(time-dependent fields), where the static isoperimetric bound should serve as a sharp upper limit under slow modulations.

In sum, we have established sharp geometric bounds at a fixed spectral gap and endowed them with quantitative rigidity, paving 
the way for precise metrological certification (fingerprinting) of traps and for extensions to many-body and field-theoretic 
models.

Here we briefly compare our sharp static bound and rigidity with kindred estimates—from the TRK sum rule and spectral 
inequalities to geometric QSLs. The aim is to stress that the scaffold (TRK + a ''Chebyshev-type'' bound via the active gap) 
is known, whereas the new content is the \textbf{sharpness + full iff + quantitative rigidity} (D1-D3) together with 
the magnetic extensions.

\begin{table}[H]
\centering
\small
\begin{tabular}{|p{3.4cm}|p{3.5cm}|p{2.8cm}|p{2.6cm}|p{3.6cm}|}
\hline
\textbf{Class of estimate} & \textbf{Typical formulation} & \textbf{Gap dependence} & \textbf{Sharpness} & \textbf{Positioning of our work} \\
\hline
TRK ($f$-sum) for $x$ \cite{SakuraiQM} &
$\sum_{n>0}(E_n{-}E_0) |x_{n0}|^2=\hbar^2/(2m)$ &
Does not directly yield $\Var_0(x)$; a ''Chebyshev'' step via the active gap is needed &
Identity (sharp for the sum), but no criterion for $\Var_0$ &
We turn TRK{+}gap into a \textbf{sharp} bound $\Var_0(x)\le \hbar^2/(2m\Delta)$ with a \textbf{full iff} (harmonic case) \\
\hline
''Chebyshev-type'' bound via the active gap (general form for $A$) &
$\Var_0(A)\le \dfrac{1}{2\Delta_A}\langle [A,[H,A]]\rangle_0$ (see §\ref{sec:setup}) &
Explicitly calibrated by the active gap $\Delta_A$ &
Sharp in abstract form, but without an equality characterization for $A=x$ and without stability &
We provide a transparent iff for $A{=}x,x_\parallel$ and quantitative stability (D1-D3) \\
\hline
Spectral/Poincaré inequalities (various forms; cf. \cite{ReedSimonII}) &
Relate variances to the generator: a variational lower bound on the first gap, or $ f{-}\langle f\rangle ^2 \le \lambda^{-1}\langle f,(-\mathcal L)f\rangle$ &
Rely on the global spectrum; do not give an upper bound on $\Var_0(x)$ at fixed $\Delta$ &
Sharp for their own functionals, but do not address a static QSL in $x$ &
Our result is precisely an \textbf{upper} bound on $\Var_0(x)$ at fixed $\Delta$, with iff and rigidity \\
\hline
Temple-type bounds (variational upper/lower bounds on levels; cf. \cite{ReedSimonII}) &
A posteriori bounds on \(E_0,E_1\) via trial states/variances &
Indirectly via estimates of \(E_1{-}E_0\) &
Sharp with good trials, but no direct control of \(\Var_0(x)\) &
We work at fixed \(\Delta\) and obtain a \textbf{universal} bound for \(\Var_0(x)\) \\
\hline
MT/QSL geometry \cite{MandelstamTamm1945,AnandanAharonov1990,DeffnerCampbell2017} &
Speed limits: \(\tau \ge \dfrac{\hbar \arccos F}{2\Delta E}\); \(g_{xx}=\Var_0(x)/\hbar^2\) for shift families \cite{ProvostVallee1980,BraunsteinCaves1994} &
Through energy spread/gaps, but dynamical &
Sharp for evolution time; no static iff characterization &
We give a \textbf{static} MT analogue: \(g_{xx}\le (2m\Delta)^{-1}\) with iff and stability \\
\hline
HPT/Kohn theorem (magnetic field) \cite{Kohn1961,Dobson1994HPT} &
Harmonic ''Kohn mode'' does not mix; dipole responses in a parabolic trap &
Does not formulate a static bound on \(\Var_0\) &
Asserts \emph{separation} under parabolicity; no general iff &
Longitudinal \textbf{iff} for \(x_\parallel\) even for inhomogeneous \(B\parallel \hat{\mathbf b}\); transverse—exact TRK for \(R_u\) and stability (§\ref{sec:magnetic-parallel}, \ref{sec:magnetic-transverse}, \ref{sec:magnetic-parallel-inhom}) \\
\hline
Polarizability/metric (static) \cite{PezzeSmerzi2018RMP,Leibfried2003RMP,Aspelmeyer2014RMP} &
\(\alpha(0)=2\sum_{n>0}\dfrac{|x_{n0}|^2}{E_n-E_0}\),\quad \(g_{xx}=\Var_0(x)/\hbar^2\) &
Depends on the full transition tail &
Bounds are typically non-sharp without specific models &
We obtain \textbf{sharp} upper bounds \(\alpha(0)\le \hbar^2/(m\Delta^2)\), \(g_{xx}\le (2m\Delta)^{-1}\) with iff; plus stability (D1-D3) \\
\hline
\end{tabular}
\caption{Where our result fits. The scaffold (TRK + gap) is known; the new ingredients are \textbf{sharpness + iff} and 
\textbf{quantitative rigidity} (D1-D3), together with magnetic counterparts featuring longitudinal iff and transverse TRK 
identities.}
\label{tab:positioning}
\end{table}

\paragraph{Key distinctions of our work.}
(i) \textbf{Sharp static QSL} for \(x\) (and \(x_\parallel\)) with a \textbf{full iff}: harmonic \(\Leftrightarrow\) saturation 
( \Cref{thm:iff-par,thm:iff-par-inhomB}). 
(ii) \textbf{Quantitative rigidity}: the deficit \(\varepsilon\) linearly controls the TRK tail via the second gap \(\Gamma\) 
and the structural norm \(V'-m\omega^2x\) (D1-D3 in §\ref{sec:rigidity}). 
(iii) \textbf{Magnetic extensions}: an exact longitudinal identity \( [H,[H,x_\parallel]]=(\hbar^2/m)\partial_\parallel V \) and iff; 
transversely—TRK for the guiding-center \(R_u\) with explicit constants (§§\ref{sec:magnetic-parallel}-\ref{sec:magnetic-transverse}). 
(iv) \textbf{Applied corollaries}: sharp upper bounds for the static polarizability and quantum metric, and a ''corridor'' for \(\Var_0(p)\) 
(§\ref{sec:conseq}).

\bibliographystyle{iopart-num} 
\bibliography{refs}

\appendix

\section{Proofs of identities and domain issues}\label{app:identities}

This appendix collects complete proofs of the identities used in the main text and brief domain remarks for the one-dimensional Hamiltonian
\(
H=\frac{p^{2}}{2m}+V(x)
\)
with a confining potential \(V\in C^{2}(\mathbb R)\), \(V(x)\to+\infty\) as \(|x|\to\infty\).
Under these assumptions \(H\) is self-adjoint on its natural domain, its spectrum is discrete and simple, and the ground state is non-degenerate and can be chosen strictly positive \(\psi_{0}(x)>0\) (see classical Sturm-Liouville results; cf. \cite{ReedSimonII}). Positivity allows one, whenever it occurs, to legitimately divide identities of the form \(F(x)\psi_{0}(x)=G(x)\psi_{0}(x)\) by \(\psi_{0}(x)\), yielding \(F=G\) almost everywhere and, by continuity, everywhere.

\begin{lemma}[Spectral identity]\label{lem:app-spectral}
Let \(\{|n\rangle\}_{n\ge0}\) be a complete set of eigenvectors with \(H|n\rangle=E_{n}|n\rangle\), \(E_{0}<E_{1}\le E_{2}\le\cdots\).
For any operator \(A\) such that \([H,A]\) is defined on \(|0\rangle\), one has
\begin{equation}\label{eq:app-spectral-identity}
\langle0|[A,[H,A]]|0\rangle
=2\sum_{n>0}(E_{n}-E_{0}) |\langle n|A|0\rangle|^{2}.
\end{equation}
\end{lemma}

\begin{proof}
Inserting the resolution of the identity \(\mathbb 1=\sum_{n}|n\rangle\langle n|\) and using \(H|n\rangle=E_{n}|n\rangle\), we obtain
\[
\langle0|[A,[H,A]]|0\rangle
=\sum_{n}\langle0|A|n\rangle\langle n|[H,A]|0\rangle
-\sum_{n}\langle0|[H,A]|n\rangle\langle n|A|0\rangle.
\]
But \(\langle n|[H,A]|0\rangle=(E_{n}-E_{0})\langle n|A|0\rangle\), hence
\[
\langle0|[A,[H,A]]|0\rangle
=\sum_{n}(E_{n}-E_{0}) \langle0|A|n\rangle\langle n|A|0\rangle
+\sum_{n}(E_{n}-E_{0}) \langle0|A|n\rangle\langle n|A|0\rangle,
\]
which yields \eqref{eq:app-spectral-identity}. The \(n=0\) term vanishes.
\end{proof}

\begin{lemma}[Double commutator for the coordinate]\label{lem:app-double}
For \(H=\frac{p^{2}}{2m}+V(x)\) the identities
\begin{equation}\label{eq:app-xx-comm}
[x,[H,x]]=\frac{\hbar^{2}}{m},
\qquad
[H,[H,x]]=\frac{\hbar^{2}}{m} V'(x)
\end{equation}
hold.
\end{lemma}

\begin{proof}
From \([x,p]=i\hbar\) and bilinearity of the commutator:
\([H,x]=\frac{1}{2m}[p^{2},x]=\frac{1}{2m}(p[p,x]+[p,x]p)=-(i\hbar/m) p\).
Then
\([x,[H,x]]=[x,-(i\hbar/m)p]=-(i\hbar/m)[x,p]=\hbar^{2}/m\).
Further,
\([H,[H,x]]=-(i\hbar/m)[H,p]=-(i\hbar/m)[V(x),p]=(\hbar^{2}/m)V'(x)\),
where we used \([p^{2},p]=0\) and \([V(x),p]=i\hbar V'(x)\).
\end{proof}

Substituting \(A=x\) into \eqref{eq:app-spectral-identity} and using \eqref{eq:app-xx-comm}, we obtain
\begin{equation}\label{eq:app-TRKx}
\sum_{n>0}(E_{n}-E_{0}) |\langle n|x|0\rangle|^{2}=\frac{\hbar^{2}}{2m}.
\end{equation}

\begin{remark}[On positivity of \(\psi_{0}\) and simplicity of the ground level]\label{rem:app-positivity}
In one dimension, for confining \(V\), the Schrödinger operator reduces to a Sturm-Liouville problem.
The ground eigenvalue is simple and the eigenfunction \(\psi_{0}\) has no nodes and can be chosen strictly positive.
This justifies division by \(\psi_{0}(x)\) in identities of the form
\([H,[H,x]]\psi_{0}=\Delta^{2}x\psi_{0}\) appearing in the main text.
See, e.g., \cite{ReedSimonII} for details.
\end{remark}

\end{document}